\documentclass[amsthm,titlepage,secthm]{elsart}

\usepackage{latexsym}              
\usepackage{amssymb}               
\usepackage{amsmath}               

\usepackage{amsfonts}

\usepackage{yjsco}
\usepackage[noend]{algorithmic}
\usepackage{float}

\usepackage[round,longnamesfirst]{natbib}
\usepackage{url}
\usepackage{enumitem}

\usepackage{graphicx}
\usepackage{color}
\usepackage{comment}

\usepackage{savesym}
\savesymbol{AND}

\setitemize[1]{leftmargin=17.5pt}

\usepackage[pdfauthor={Mark Giesbrecht and Myung Sub Kim},
            pdftitle={Computing the Hermite form of a Matrix of
              Ore Polynomials},
            colorlinks,linkcolor=DarkBlue,citecolor=DarkGreen,urlcolor=DarkRed]{hyperref}

\definecolor{DarkGreen}{rgb}{0,.35,0}
\definecolor{DarkBlue}{rgb}{0,0,.5}
\definecolor{DarkRed}{rgb}{.6,0,0}

\floatstyle{ruled} \newfloat{algg}{H}{loa}
\floatname{algg}{Algorithm}

\renewenvironment{alg}[2][nolabel]%
{\renewcommand{\thealgg}{{\tt {#2}}}
\begin{algg}
\ifthenelse{\equal{#1}{nolabel}}{\caption{ }}{\caption{ \label{#1}}}
\begin{algorithmic}[1]}%
{\end{algorithmic}\end{algg}}

\newcommand{\rank}{\textup{rank}\,}

\DeclareMathAlphabet{\mathitbf}{OML}{cmm}{b}{it}

\makeatletter
\renewenvironment{proof}%
  {\par\addvspace{\@bls \@plus 0.5\@bls \@minus 0.1\@bls}\noindent
   {\bfseries\Elproofname}\enspace\ignorespaces}%
  {\qed\par\addvspace{\@bls \@plus 0.5\@bls \@minus 0.1\@bls}}
\newenvironment{proof*}%
  {\par\addvspace{\@bls \@plus 0.5\@bls \@minus 0.1\@bls}\noindent
   {\bfseries\Elproofname}\enspace\ignorespaces}%
  {\par\addvspace{\@bls \@plus 0.5\@bls \@minus 0.1\@bls}}
\makeatother

\newcommand{\FD}{{\F[\partial;\sigma,\delta]}}

\newcommand{\qFD}{{\F(\partial;\sigma,\delta)}}
\DeclareMathOperator{\Ddet}{{\mathcal Det}}
\DeclareMathOperator{\PDdet}{{\delta\varepsilon\tau}}
\DeclareMathOperator{\sign}{{\mathord{\rm sign}}}

\DeclareMathOperator{\M}{\textsf{M}}
\newcommand{\D}{{\partial}}

\newcommand{\QQ}{{\mathbb{Q}}}
\newcommand{\NN}{{\mathbb{N}}}
\newcommand{\ZZ}{{\mathbb{Z}}}
\newcommand{\F}{{\mathsf{F}}}

\newcommand{\K}{{\mathsf{K}}}
\renewcommand{\k}{{\mathsf{k}}}
\newcommand{\shift}{{\mathcal{S}}}
\newcommand{\diff}{{\lower3pt\hbox{\large$'$}}}
\newcommand{\degD}{{\deg_{\D}}}

\newcommand{\nxn}{{n\times n}}

\newcommand{\nx}[1]{{n\times #1}}

\newcommand{\lclm}{{\mbox{\upshape lclm}}}
\newcommand{\lcrm}{{\mbox{\upshape lcrm}}}
\newcommand{\gcld}{{\mbox{\upshape gcld}}}
\newcommand{\gcrd}{{\mbox{\upshape gcrd}}}

\newcommand{\Hbar}{{\bar H}}
\newcommand{\gbar}{{\bar g}}
\newcommand{\vbar}{{\bar v}}
\newcommand{\ubar}{{\bar u}}
\newcommand{\undef}{\perp}
\newcommand{\softO}{{O\,\tilde{}\,}}
\newcommand{\diag}{{\mbox{\upshape diag}}}
\newcommand{\row}{{\mbox{\upshape row}}}
\renewcommand{\L}{{\mathcal L}}

\newcommand{\Glin}{{\widehat{G}}}
\newcommand{\Glinred}{{\widetilde{G}}}
\newcommand{\Alin}{{\widehat{A}}}
\newcommand{\Alinred}{{\widetilde{A}}}

\newcommand{\Tlin}{{\widehat{T}}}
\newcommand{\inorm}[1]{{\|#1\|_\infty}}
\newcommand{\rnorm}[1]{{\|#1\|_\infty^{\mskip-2mu(\varrho)}}}
\newcommand{\nonzero}{{\setminus\{0\}}}

\newcommand{\ssty}{\scriptstyle}

\numberwithin{equation}{section}
\numberwithin{table}{section}

\begin{document}

\begin{frontmatter}

  \title{Computing the Hermite Form of a Matrix of Ore
    Polynomials}

    \author{Mark Giesbrecht}
    \address{Cheriton School of Computer Science, University of Waterloo,
      Waterloo, ON, Canada}
    \ead{mwg@uwaterloo.ca}
    
    \author{Myung Sub Kim}
    \address{Cheriton School of Computer Science, University of Waterloo,
      Waterloo, ON, Canada}
    \ead{ms2kim@uwaterloo.ca}
    
    
    \begin{abstract}
      Let $\FD$ be the ring of Ore polynomials over a field (or a skew
      field) $\F$, where $\sigma$ is an automorphism of $\F$ and
      $\delta$ is a $\sigma$-derivation.  Given a matrix
      $A\in\FD^{m\times n}$, we show how to compute the Hermite form
      $H$ of $A$ and a unimodular matrix $U$ such that $UA=H$.  The
      algorithm requires a polynomial number of operations in $\F$ in
      terms of the dimensions $m$ and $n$, and the degrees (in $\D$) of the
      entries in $A$.  When $\F=\k(z)$ for some field $\k$, it also
      requires time polynomial in the degrees in $z$ of the
      coefficients of the entries, and if $\k=\QQ$
      it requires time polynomial in the bit length of the
      rational coefficients as well.  Explicit analyses are provided for the
      complexity, in particular for the important cases of
      differential and shift polynomials over $\QQ(z)$.  To accomplish
      our algorithm, we apply the Dieudonn\'e determinant and
      quasideterminant theory for Ore polynomial rings to get explicit
      bounds on the degrees and sizes of entries in $H$ and $U$.
    \end{abstract}
    
  \end{frontmatter}
  
  \section{Introduction}

The Ore polynomials are a natural algebraic structure which captures
difference, $q$-difference, differential, and other non-commutative
polynomial rings.  The basic concepts of pseudo-linear algebra are presented
nicely by \cite{BronsteinPetkovsek:1996}; see \citep{Ore31} for the
seminal introduction.

On the other hand, canonical forms of matrices over commutative
principal ideal domains (such as $\ZZ$ or $\F[x]$, for a field $\F$)
have proven invaluable for both mathematical and computational
purposes. One of the successes of computer algebra over the past three
decades has been the development of fast algorithms for computing
these canonical forms. These include triangular forms such as the
Hermite form \citep{Her51}, low degree forms like the Popov
form \citep{Popov:1972}, as well as the diagonal Smith form
\citep{Smith:1861}.

Canonical forms of matrices over non-commutative domains, especially
rings of differential and difference operators, are also extremely
useful. These have been examined at least since the work of \cite{Dickson:1923},
\cite{Wedderburn:1932}, and \cite{Jac43}.  Recently they have found
uses in control theory \citep{ChyQua05,Zer06,Hal08}.  Computations
with multidimensional linear systems over Ore algebras are nicely
developed by \cite{ChyQua07}, and an excellent implementation of many
fundamental algorithms is provided in the OreModules package of Maple.

In this paper we consider canonical forms of matrices of Ore
polynomials over a skew field $\F$. Let
$\sigma:\F\to\F$ be an automorphism of $\F$ and $\delta:\F\to\F$ be a
$\sigma$-derivation.  That is, for any $a,b\in\F$,
$\delta(a+b)=\delta(a)+\delta(b)$ and
$\delta(ab)=\sigma(a)\delta(b)+\delta(a)b$.  We then define $\FD$ as
the set of usual polynomials in $\F[\D]$ under the usual addition, but
with multiplication defined by
\[
\D a = \sigma(a) \D + \delta(a)
\]
for any $a\in\F$.  This is well-known to be a left (and right)
principal ideal domain, with a straightforward Euclidean algorithm
(see \citep{Ore33}).

Some important cases over the field of rational functions $\F=\k(z)$ over a field
$\k$ are as follows:
\begin{itemize}
\item[(1)] $\sigma(z)=\shift(z)=z+1$ is a so-called \emph{shift}
  automorphism of $\k(z)$, and $\delta$ identically zero on $\k$.
  Then $\k(z)[\D;\shift,0]$ is generally referred to as the ring of
  \emph{shift polynomials}. With a slight abuse of
  notation we write $\k(z)[\D;\shift]$ for this ring.
\item[(2)] $\delta(z)=1$ and $\sigma(z)=z$, so $\delta(h(z))=h'(z)$
  for any $h\in\k(z)$ with $h'$ its usual derivative.  Then
  $\k(z)[\D;\sigma,\delta]$ is called the ring of \emph{differential
    polynomials}.  With a slight abuse of notation we write
  $\k(z)[\D;\diff]$ for this ring.

  A primary motivation in the definition of $\k(z)[\D;\diff]$ is that
  there is a natural action on the space of infinitely differentiable
  functions in $z$, namely the differential polynomial
  \[
  a_m\D^m+a_{m-1}\D^{m-1}+\cdots + a_1\D+a_0\in\k(z)[\D;\diff]
  \]
  acts as the linear differential operator
  \[
  a_m(z)\frac{d^m f(z)}{d z^m}+a_{m-1}(z)\frac{d^{m-1}f(z)}{d
    z^{m-1}}+ \cdots+a_1(z)\frac{d f(z)}{d z}+a_0(z)f(z)
  \]
  on an infinitely differentiable function $f(z)$.  See
  \citep{BronsteinPetkovsek:1996}. 
\end{itemize}


The (row) Hermite form we will compute here is achieved purely by row
operations, and we treat a matrix $A\in\FD^{m\times n}$ as
generating the left $\FD$-module of its rows.  Thus, by \emph{left row
  rank}, we mean the rank of the free left $\FD$-module of rows of $A$, and
will denote this simply as the \emph{rank} of $A$ for the remainder of
the paper.  A matrix $H\in\FD^{m\times n}$ of rank $r$ is in \emph{Hermite
form} if an only if
\begin{itemize}
\item[\emph{(i)}] Only the first $r$ rows are non-zero;
\item[\emph{(ii)}] In each row the leading (first non-zero) element is monic;
\item[\emph{(iii)}] All entries in the column below the leading element in
  any row are zero;
\item[\emph{(iv)}] All entries in the column above the leading element in any
  row are of lower degree than the leading element.
\end{itemize}
For square matrices of full rank the Hermite form will thus be upper
triangular with monic entries on the diagonal, whose degrees dominate
all other entries in their column.

For example, in the differential polynomial ring $\QQ(z)[\D;\diff]$ as
above:
\begin{equation}
\label{eq:exA}
A=
\begin{pmatrix}
{\ssty 1}+{\ssty (z+2)}\D+\D^{2} & {\ssty 2}+{\ssty (2z+1)}\D & {\ssty
  1}+{\ssty (1+z)}\D\\
{\ssty (2z+z^{2})}+{\ssty z}\D & {\ssty (2+2z+2z^{2})}+\D & {\ssty 4z+z^{2}}\\
{\ssty (3+z)}+{\ssty (3+z)}\D+\D^{2} & {\ssty (8+4z)}+{\ssty
  (5+3z)}\D+\D^{2} & {\ssty (7+8z)}+{\ssty (2+4z)}\D
\end{pmatrix}
\in\QQ(z)[\D;\diff]^{3\times 3}
\end{equation}
has Hermite form
\[
H=
\begin{pmatrix}
{\ssty (2+z)}+\D & {\ssty 1+2z} & \frac{-2+z+2z^{2}}{2z}-\frac{1}{2z}\D\\
0 & {\ssty (2+z)}+\D & {\ssty 1}+\frac{7z}{2}+\frac{1}{2}\D\\
0 & 0 & -\frac{2}{z}+\frac{-1+2z+z^{2}}{z}\D+\D^{2}
\end{pmatrix}
\in\QQ(z)[\D;\diff]^{3\times 3}.
\]
Note that the Hermite form may have denominators in $z$.  Also, while
this example does not demonstrate it, the degrees in the Hermite form,
in both numerators and denominators in $z$ and $\D$, are generally
substantially larger than in the input (in Theorem
\ref{thm:sqhermcost} we will provide polynomial, though quite large,
bounds on these degrees, and suspect these bounds may well be met
generically).

For any matrix $A\in\FD^\nxn$ of full rank, there exists a unique
unimodular matrix $U\in\FD^\nxn$ (i.e., a matrix whose inverse exists
and is also in $\FD^\nxn$) such that $UA=H$ is in Hermite form. This
form is canonical in the sense that if two matrices $A,B\in\FD^\nxn$
are such that $A=PB$ for unimodular $P\in\FD^\nxn$ then the Hermite
form of $A$ equals the Hermite form of $B$.  Existence and uniqueness
of the Hermite form are established much as they are over $\ZZ^\nxn$
in Section \ref{sec:exuniq}.  For rank deficient matrices and
rectangular matrices, the Hermite form also exists but the
transformation matrix may not be unique.  See Section
\ref{sec:oddshapes} for further details.

In commutative domains such as $\ZZ$ and $\F[x]$ there have been
enormous advances in the past two decades in computing Hermite, Smith
and Popov forms.  Polynomial-time algorithms for the Smith and Hermite
forms over $\F[x]$ were developed by \cite{Kannan:1985}, with
important advances by \cite{KalKri87}, \cite{Villard:1995},
\cite{MuldersStorjohann:2003}, \cite{PerSte10}, and many others.  One of the key
features of this recent work in computing canonical forms has been a
careful analysis of the complexity in terms of matrix size, entry
degree, and coefficient swell. Clearly identifying and analyzing the
cost in terms of all these parameters has led to a dramatic drop in
both theoretical and practical complexity.

Computing the classical Smith and Hermite forms of matrices over Ore
domains has received less attention though canonical forms of
differential polynomial matrices have applications in solving
differential systems and control theory (see
\citep{Hal08,KotLei08}). \cite{AbramovBronstein:2001} analyze the
number of reduction steps necessary to compute a \emph{row-reduced
  form}, while \linebreak
\cite{BeckermannChengLabahn:2006} analyze the
complexity of row reduction in terms of matrix size, degree and the
sizes of the coefficients of some shifts of the input
matrix. \cite{BeckermannChengLabahn:2006} demonstrate tight bounds on
the degree and coefficient sizes of the output, which we will employ
here. For the Popov form, \cite{cheng:phd} gives an algorithm for
matrices of shift polynomials. Cheng's approach involves order bases
computation in order to eliminate lower order terms of Ore polynomial
matrices. A main contribution of \cite{cheng:phd} is to give an
algorithm computing the rank and a row-reduced basis of the left
nullspace of a matrix of Ore polynomials in a fraction-free way. This
idea is extended in \cite{DaviesChengLabahn:2008} to compute the Popov
form of general Ore polynomial matrices.  They reduce the problem of
computing Popov form to a nullspace computation. However, though Popov
form is useful for rewriting high order terms with respect to low
order terms, we want a different canonical form more suited to solving
system of linear diophantine equations. Since the Hermite form is
upper triangular, it meets this goal nicely, not to mention the fact
that it is a ``classical'' canonical form.  An implementation of the
basic (exponential-time) Hermite algorithm is provided by
\cite{Cul05}.  In \citep{GieKim09}, we present a polynomial-time
algorithm for the Hermite form over $\QQ(z)[\D;\diff]$, for full rank
square matrices.  While it relies on similar techniques as this
current paper, the cost of the algorithm is higher,
the coefficient bounds weaker, and it does not work for matrices of
general Ore polynomials.

The related ``two-sided'' problem of computing the Jacobson
(non-commutative Smith) canonical form has also been recently
considered. \cite{BliCid03} implement the standard algorithm in the
package \emph{Janet}.  \cite{LevSch11} provide a very complete
implementation, for the full Ore case over skew fields, of a Jacobson
form algorithm using Gr\"obner bases in Singular.  \cite{Middeke:2008}
has recently demonstrated that the Jacobson form of a matrix of
differential polynomials can be computed in time polynomial in the
matrix size and degree (but the coefficient size is not analyzed).
\cite{GieHei12} give a probabilistic polynomial-time algorithm for
this problem in the differential case.

One of the primary difficulties in both developing efficient
algorithms for matrices of Ore polynomials, and in their analysis, is
the lack of a standard notion of determinant, and the important bounds
this provides on degrees in eliminations.  In Section
\ref{sec:detbounds} we establish bounds on the degrees of entries in
the inverse of a matrix over any non-commutative field with a
reasonable degree function.  We do this by introducing the
\emph{quasideterminant} of \cite{GelRet91,GelRet92} and analyzing its
interaction with the degree function.  We also prove similar bounds on
the degree of the Dieudonn\'e determinant.  In both cases, the bounds
are essentially the same as for matrices over a commutative function
field.

In Section \ref{sec:orebounds} we consider matrices over the Ore
polynomials and bound the degrees of entries in the Hermite form and
corresponding unimodular transformation matrices.  We also bound the
degrees of the Dieudonn\'e determinants of these matrices.

In Section \ref{sec:compute} we present our algorithm for the Hermite
form. The degree bounds and costs of our algorithms are summarized as
follows, from Theorems \ref{thm:Hdeg}, \ref{thm:HermDegV} and
\ref{thm:sqhermcost}.

\newtheorem*{thma}{Summary Theorem}

\begin{thma}
  Let $A\in\k[z][\D;\sigma,\delta]^\nxn$ have full rank with entries
  of degree at most $d$ in $\D$, and coefficients of degree at most
  $e$ in $z$.  Let $H\in\k(z)[\D;\sigma,\delta]^\nxn$ be the Hermite
  form of $A$ and $U\in\k(z)[\D;\sigma,\delta]^\nxn$ such that $UA=H$.
  \begin{enumerate}
  \item[(a)] The sum of degrees in $\D$ in any row of $H$ is at most
    $nd$, and each entry in $U$ has degree in $\D$ at most $(n-1)d$.
  \item[(b)] All coefficients from $\k(z)$ of entries of $H$ and $U$
    have degrees in $z$, of both numerators and denominators, bounded
    by $O(n^2de)$.
  \item[(c)] We can compute $H$ and $U$ deterministically with
    $\softO(n^9d^4e)$%
    \footnote[2]{We employ soft-Oh notation: for functions $\sigma$
      and $\varphi$ we say $\sigma\in \softO(\varphi)$ if $\sigma\in
      O(\varphi \log^c\varphi)$ for some constant $c\geq 0$.},
    operations in $\k$.
  \item[(d)] Assume $\k$ has at least $4n^2de$ elements.  We can
    compute the Hermite form $H$ and $U$ with an expected number of
    $\softO(n^7d^3e)$ of operations in $\k$ using standard
    polynomial arithmetic.  This algorithm is probabilistic of the Las
    Vegas type; it never returns an incorrect answer.
    \end{enumerate}
\end{thma}
The cost of our algorithm for Ore polynomials over an arbitrary skew
field, as well as over more specific fields like $\QQ(z)$, is also
shown to be polynomially bounded, and is discussed in Section
\ref{sec:compute}.  

It should be noted from the above theorem that the output is of quite
substantial size.  The transformation matrix $U$ as above is an
$n\times n$ matrix of polynomials in $\D$ of degree bounded by $nd$ in
$\D$ and each coefficient has degree bounded by $n^2de$, for a total
size of $O(n^5d^2e)$ elements of $\k$.  While we have not proven our
size bounds are tight, we have some confidence they are quite strong.

The algorithm presented in Section \ref{sec:compute} is derived from
the ``linear systems'' approach of \cite{KalKri87} and \cite{Sto94}.
In particular, it reduces the problem to that of linear system solving
over the generally commutative ground field (e.g., $\k(z)$).  There
are efficient algorithms and implementations for solving this problem.
While we expect that further algorithmic refinements and reductions in
cost can be achieved before an industrial-strength implementation is
made, the general approach of reducing to well-studied computational
problem in a commutative domain would seem to have considerable merit
in theory and practice.

In Section \ref{sec:oddshapes} we show that for the case of
rank-deficient and rectangular matrices, the computation of the
Hermite form is reduced to the full rank, square case.


  \section{Existence and Uniqueness of the Hermite form over Ore domains}
\label{sec:exuniq}

In this section we establish the basic existence and uniqueness of
Hermite forms over Ore domains. These follow similarly to the
traditional proofs over $\ZZ$; see for example \citep[Theorems II.2
and II.3]{New72}, which we outline below.

\begin{fact}[\cite{Jac43}, Section 3.7]
  \label{fact:2x2uni}
  Let $a,b\in\FD$, not both zero with $g=\gcrd(a,b)$, $u,v\in\FD$ such
  that $ua+vb=g$, and $s,t\in\FD$ such that $sa=-tb=\lclm(a,b)$.
  Then 
  \[
  W=\begin{pmatrix}
    u & v\\
    s & t
  \end{pmatrix}
  \in\FD^{2\times 2}
  ~~\mbox{is such that}~~
  W
  \begin{pmatrix}
    a\\ b
  \end{pmatrix}
  =
  \begin{pmatrix}
    g\\ 0
  \end{pmatrix},
  \]
  and $W$ is unimodular.
\end{fact} 

This is easily generalized to $n\times n$ matrices as follows.
\begin{lem}
  \label{lem:nxnuni}
  Let $w=(w_1,\ldots,w_n)^T\in\FD^\nx1$, and
  $i,j\in\{1,\ldots,n\}$. There exists a matrix
  \[
  E=E(i,j; w)\in\FD^\nxn
  \]
  such that $Ew=(u_1,\ldots,u_n)^T\in\FD^\nx1$, with
  $u_i=\gcrd(w_i,w_j)$ and $u_j=0$.
\end{lem}
\begin{proof}
  If both $w_i=w_j=0$ are zero, then $E$ is the identity matrix.  If
  $w_i=0$ and $w_j\neq 0$, then let $E$ be the permutation matrix
  which swaps rows $j$ and $i$.

  Otherwise, let $W=\begin{pmatrix} u & v\\ s & t \end{pmatrix}$ be as
  in Fact \ref{fact:2x2uni} with $(a,b)^T=(w_i,w_j)^T$, and
  $W(w_i,w_j)^T=(g,0)^T$ for $g=\gcrd(w_i,w_j)$.  Define $E(i,j; w)$
  as the identity matrix except
  \[
  E_{ii}=u, ~~~ E_{ij}=v, ~~~ E_{ji}=s, ~~~E_{jj}=t.
  \]
  Clearly $E$ satisfies the desired properties.
\end{proof}

We note that off diagonal entries in a triangular matrix can be
unimodularly reduced by the diagonal entry below it.

\begin{lem}
  \label{lem:unired}
  Let $J\in\FD^\nxn$ be upper triangular with non-zero diagonal. There
  exists a unimodular matrix $R\in\FD^\nxn$, which is upper triangular
  and has ones on the diagonal, such that in every column of $RJ$, the
  degree of each diagonal entry is strictly larger than the degrees of
  the entries above it.
\end{lem}
\begin{proof}
  For any $a,b\in\FD$ with $b\neq 0$, we have $a=qb+r$ for quotient
  $q\in\FD$ and remainder $r\in\FD$ with $\degD r<\degD b$, and
  \[
  \begin{pmatrix}
    1 & -q \\
    0 & 1
  \end{pmatrix}
  \begin{pmatrix}
    a \\ b 
  \end{pmatrix}
  = 
  \begin{pmatrix}
    r\\ b
  \end{pmatrix}.
  \]
  Embedding such unimodular matrices $Q$ into $n\times n$ identity
  matrices, we can ``reduce'' the off diagonal entries of $J$ by the
  diagonal entries below them.
\end{proof}

\begin{thm} 
  \label{thm:sqherm}
  Let $A\in\FD^{n\times n}$ have full rank. Then there exists a matrix
  $H\in\FD^\nxn$ in Hermite form, and a unimodular matrix
  $U\in\FD^\nxn$, such that $UA=H$.
\end{thm} 
\begin{proof} 
  The proof follows by observing the traditional (but inefficient)
  algorithm to compute the Hermite form.  We first use a (unimodular
  row) permutation to move any non-zero element in column 1 into the
  top left position; failure to find a non-zero element in column 1
  means our matrix is rank deficient.  We then repeatedly apply Lemma
  \ref{lem:nxnuni} to find $Q_1$ such that $Q_1A$ only has the top
  left position non-zero.  This same procedure is then repeated on
  subdiagonal of columns $2, 3, \ldots, n$ in sequence, so there
  exists a unimodular matrix $Q=Q_1\cdots Q_n$ such that $QA$ is upper
  triangular.  The matrix is then unimodularly reduced using Lemma
  \ref{lem:unired}.
\end{proof}

\begin{thm} 
  \label{thm:hermuniq}
  Let $A\in\FD^\nxn$ have full row rank.  Suppose $UA=H$ for unimodular
  $U\in\FD^\nxn$ and Hermite form $H\in\FD^\nxn$.  Then both $U$ and $H$
  are unique.
\end{thm}
\begin{proof}
  Suppose $H$ and $G$ are both Hermite forms of $A$. Thus, there exist
  unimodular matrices $U$ and $V$ such that $UA=H$ and $VA=G$, and
  $G=WH$ where $W=VU^{-1}$ is unimodular.  Since $G$ and $H$ are upper
  triangular matrices, we know $W$ is as well.  Moreover, since $G$
  and $H$ have monic diagonal entries, the diagonal entries of $W$
  equal $1$. We now prove $W$ is the identity matrix. By way of
  contradiction, first assume that $W$ is not the identity, so there
  exists an entry $W_{ij}$ which is the first nonzero off-diagonal
  entry on the $i$th row of $W$. Since $i<j$ and since $W_{ii}=1$,
  $G_{ij}=H_{ij}+W_{ij}H_{jj}$. Because $W_{ij}\neq 0$, we see $\degD
  G_{ij}\geq\degD G_{jj}$, which contradicts the definition of the
  Hermite form.  

  Uniqueness of $U$ is easily established since $UA=VA$, so
  $U^{-1}V=I$ and $U=V$.
\end{proof}


  \section{Non-commutative determinants and degree bounds for linear
  equations}
\label{sec:detbounds}

One of the main difficulties in matrix computations in skew (non-commutative)
fields, and a primary difference with the commutative case, is the
lack of the usual determinant.  In particular, the determinant allows
us to bound the degrees of solutions to systems of equations, the size
of the inverse or other decompositions, not to mention the degrees at
intermediate steps of computations, through Hadamard-like formulas and
Cramer's rules.

The most common non-commutative determinant was defined by
\cite{Dieudonne:1943}, and is commonly called the \emph{Dieudonn\'e
  determinant}.  It preserves some of the multiplicative properties of
the usual commutative determinant, but is insufficient to establish
the degree bounds we require (amongst other inadequacies). 
\cite{GelRet91,GelRet92} introduced \emph{quasideterminants} and a
rich associated theory as a central tool in linear algebra over
non-commutative rings.  Quasideterminants are more akin to the
(inverse of the) entries of the classical adjoint of a matrix than a
true determinant.  We employ quasideterminants here to establish
bounds on the degree of the entries in the inverse of a matrix, and on
the Dieudonn\'e determinant in this section, and on the Hermite form
and its multiplier matrices in Section \ref{sec:orebounds}.

We will establish bounds on degrees of quasideterminants and
Dieudonn\'e determinants for a general skew
field $\K$ with a \emph{degree} $\deg:\K\to\ZZ\cup
\{-\infty\}$ satisfying the following properties.  For $a,b\in\K$:
\begin{enumerate}
\item[(i)] If $a\neq 0$ then $\deg a\in\ZZ$, and $\deg 0=-\infty$;
\item[(ii)] $\deg (a+b)\leq \max\{\deg a,\deg b\}$;
\item[(iii)] $\deg(ab)=\deg a+\deg b$;
\item[(iv)] If $a\neq 0$ then $\deg(a^{-1})=-\deg a$.
\end{enumerate}

As a simple commutative example, if $\K=\F(y)$ for some field $\F$ and commuting
indeterminate $y$, for any $a=a_N/a_D$ with polynomials $a_N,
a_D\in\F[y]$ ($a_D\neq 0$), we
can define $\deg a=\deg a_N-\deg a_D$.  

More properly, our degree function is a \emph{non-archimedean
  valuation} on $\K$. Since our main application will be to
non-commutative Ore polynomial rings, where degrees are a natural and
traditional notion, we will adhere to the nomenclature of degrees.  \emph{We
note, however, that the degrees as defined here may become negative.}
See Lemma \ref{lem:Oredeggood} for the effective application to the
Ore polynomial case.

\subsection{Quasideterminants and degree bounds}

Following \cite{GelRet91,GelRet92}, we define the
\emph{quasideterminant} as a collection of $n^2$ functions from
$\K^\nxn\to\K\cup \{\undef\}$, where $\undef$ represents the function
being \emph{undefined}.  Let $A\in\K^\nxn$ and $p,q\in\{1,\ldots,
n\}$.  Assume $A_{pq}\in\K$ is the $(p,q)$ entry of $A$,
and let $A^{(pq)}\in\K^{(n-1)\times (n-1)}$ be the matrix $A$ with the
$p$th row and $q$th column removed.  Define the
$(p,q)$-quasideterminant of $A$ as
\[
|A|_{pq}=A_{pq}- \sum_{i\neq p, j\neq q} A_{pi} (|A^{(pq)}|_{ji})^{-1} A_{jq},
\]
where the sum is taken over all summands where $|A^{(pq)}|_{ji}$ is
defined.  If all summands have $|A^{(pq)}|_{ji}$ undefined then
$|A|_{pq}$ is undefined (and has value $\undef$).  See
\citep{GelRet92}.

\begin{fact}[\cite{GelRet91}, Theorem  1.6]
  \label{ft:qdetinv}
  Let $A\in\K^\nxn$ over a (possibly skew) field $\K$.
  \begin{itemize}
  \item[(1)] The inverse matrix $B=A^{-1}\in\K^\nxn$ exists if and only if the
    following are true:
  \begin{itemize}
    \item[(a)] If the quasideterminant $|A|_{ij}$ is defined then
      $|A|_{ij}\neq 0$, for all $i,j\in\{1,\ldots,n\}$;
    \item[(b)] For all $p\in\{1,\ldots,n\}$ there exists a
    $q\in\{1,\ldots,n\}$, such that the quasideterminant $|A|_{pq}$
    is defined;
    \item[(c)] For all $q\in\{1,\ldots,n\}$ there exists a
    $p\in\{1,\ldots,n\}$ such that the quasideterminant $|A|_{pq}$
    is defined;
    \end{itemize}
  \item[(2)] If the inverse $B$ exists, then for
    $i,j\in\{1,\ldots,n\}$ we have
      \[
      B_{ji}=\begin{cases}
        (|A|_{ij})^{-1} & \mbox{if $|A|_{ij}$ is defined,}\\
        0 & \mbox{if $|A|_{ij}$ is not defined.}
      \end{cases}
      \]
    \end{itemize}
\end{fact}

Over a commutative field $\K$, where $A\in\K^\nxn$ has inverse $B$,
the quasideterminants behave like a classical adjoint:
$|A|_{ij}=(-1)^{i+j}\det A/\det A^{(ij)}=1/B_{ji}$.  If $B_{ji}$ is zero
then $|A|_{ij}$ is undefined.

We now bound the size of the quasideterminants in terms of the size of
the entries of $A$.  Assume that $\K$ has a degree function as above.

\begin{thm}
  \label{thm:qdetbd}
  Let $A\in\K^\nxn$, such that either $A_{ij}=0$ or $0\leq \deg
  A_{ij}\leq d$ for all $i,j\in\{1,\ldots,n\}$.  For all $p,q\in\{1,\ldots,n\}$
  such that $|A|_{pq}$ is defined we have $-(n-1) d\leq
  \deg |A|_{pq} \leq nd$.
\end{thm}
\begin{proof*}
  We proceed by induction on $n$.

  For $n=1$, $p=q=1$ and $|A|_{11}=A_{11}$, so clearly the property
  holds.  Assume the statement is true for dimension $n-1$.  Then
  \[
  \deg |A|_{pq} 
   = \deg\left(
    A_{pq}
    - \sum_{i\neq p, j\neq q} A_{pi}
    (|A^{(pq)}|_{ji})^{-1}  A_{jq}\right),
  \]
  where the sum is over all defined summands.  Then using the
  inductive hypothesis we have
  \begin{align*}
  \deg |A|_{pq}
     & \leq \max\left \{ \deg A_{pq},  \max_{i\neq p, j\neq q} \left\{\deg A_{pi} -
      \deg  |A^{(pq)}|_{ji} + \deg A_{jq}\right\}
  \right\}\\
   & \leq 2d+ (n-2)d  \leq nd,
  \end{align*}
  and
  \[
  \deg |A|_{pq}
   \geq - \deg  |A^{(pq)}|_{ji}\geq -(n-1)d. \qed
  \]
\end{proof*}

\begin{cor}
  \label{cor:uniinvbd}
  Let $A\in\K^\nxn$ be unimodular, and $B\in\K^\nxn$
  such that $AB=I$.   Assume $A_{ij}=0$ or $0\leq \deg A_{ij}\leq d$
  for all $i,j\in\{1,\ldots,n\}$.  Then $\deg B \leq (n-1)d$.
\end{cor}
\begin{proof}
  From  Fact \ref{ft:qdetinv} we know that $B_{ji}=(|A|_{ij})^{-1}$
  when $|A|_{ij}$ is defined (and $B_{ji}=0$ otherwise).  Thus
  $\deg B_{ji}=-\deg|A|_{ij} \leq (n-1)d$, and $B_{ij}=0$ or
  $\deg B_{ij}\geq 0$ since $A$ is unimodular.
\end{proof}

\subsection{Dieudonn\'e Determinants}
\label{ssec:dieu}

Let $[\K^*,\K^*]$ be the \emph{commutator subgroup} of the
multiplicative group $\K^*$ of $\K$, the (normal) subgroup of
$\K^*$ generated by all pairs of elements of the form $a^{-1}b^{-1}ab$
for $a,b\in\K^*$.  Thus $\K^*/[\K^*,\K^*]$ is a commutative
group.  

Let $A\in\K^{n\times n}$ be a matrix with a right inverse.  The
\emph{Bruhat Normal Form} of $A$ is a decomposition $A=TDPV$, where
$P\in\K^\nxn$ is a permutation matrix inducing the permutation
$\sigma:\{1,\ldots,n\}\to\{1,\ldots,n\}$, and $T,D,V\in\K^\nxn$ are
\[
T=
\begin{pmatrix}
  1 & * & \cdots & * \\
  0 & 1 & \cdots & *  \\
  \vdots &  & \ddots & \vdots\\
  0 & \cdots & 0 & 1
\end{pmatrix},
~~~~~
D = \diag(u_1,\ldots,u_n),
~~~~~
V=
\begin{pmatrix}
  1 & 0 & \cdots & 0 \\
  * & 1 & \cdots & 0  \\
  \vdots &  & \ddots & \vdots\\
 * & \cdots & * & 1
\end{pmatrix}.
\]
See \citep[Chapter 19]{Dra83} for more details.
The Bruhat decomposition arises from Gaussian elimination, much as 
the $LUP$ decomposition does in the commutative case.
We then define $\PDdet(A)=\sign(\sigma)\cdot u_1\cdots u_n\in\K$ (sometimes
called the \emph{pre-determinant} of $A$).    Let $\pi$ be the canonical
projection from $\K^*\to\K/[\K^*,\K^*]$.  Then the Dieudonn\'e
determinant is defined as $\Ddet(A)=\pi(\PDdet(A))\in\K/[\K^*,\K^*]$, or
$\Ddet(A)=0$ if $A$ is not invertible.

The Dieudonn\'e determinant has a number of the desirable properties
of the usual determinant, as proven in \citep{Dieudonne:1943}:

\begin{enumerate}
\item $\Ddet(AB)=\Ddet(A)\Ddet(B)$ for any $A,B\in\K^\nxn$;
\item $\Ddet(P)=1$ for any permutation matrix;
\item $\Ddet \begin{pmatrix} A & C \\
                                            0 & B \end{pmatrix}
          = \Ddet(A)\Ddet(B)$.
\end{enumerate}

Also note that if $\K$ has a degree function as above, then 
$\deg(\Ddet(A))$ is well defined, since all elements of the
equivalence class of $\pi(\Ddet(A))$ have the same degree (since the
degree of all members of the commutator subgroup is zero).
\cite{GelRet91} show that 
\begin{align*}
\PDdet(A) & = |A|_{11} |A^{(11)}|_{22} |A^{(12,12)}|_{33}
|A^{(123,123)}|_{44} \cdots |A^{(1\ldots n-1,1\ldots,n-1)}|_{nn} \\
& = |A|_{11} \cdot \PDdet(A^{(11)}),
\end{align*}
when all these quasideterminants are defined (or equivalently $P$ is
the identity in the Bruhat decomposition above), where $A^{(1\ldots
  k,1\ldots k)}$ is the matrix $A$ with rows $1\ldots k$ and columns
$1\ldots k$ removed (keeping the original labelings of the remaining
rows and columns).

More generally, let $R=(r_1,\ldots,r_n)$, $C=(c_1,\ldots,c_n)$ be
permutations of $\{1,\ldots,n\}$, let $R_k=(r_1,\ldots,r_k)$,
$C_k=(c_1,\ldots,c_k)$, and define $A^{(R_k,C_k)}$ as the matrix $A$
with rows $r_1,\ldots,r_k$ and columns $c_1,\ldots,c_k$ removed
(where $A^{(R_0,C_0)}=A$).  Define
\begin{equation}
\label{eq:predetexp}
\begin{aligned}
\PDdet_{R,C}(A)
& = |A|_{r_1,c_1} |A^{(R_1,C_1)}|_{r_2,c_2} |A^{(R_2,C_2)}|_{r_3,c_3}
\cdots |A^{(R_{n-1},C_{n-1})}|_{r_n,c_n} \\
& = |A|_{r_1,c_1} \cdot \PDdet_{R,C}(A^{(r_1,c_1)}) \\
& = |A|_{r_1,c_1} \cdot |A^{(R_1,C_1)}|_{r_2,c_2} \cdot  \PDdet ( A^{(R_2,C_2)}).
\end{aligned}
\end{equation}

\begin{fact}[\cite{Gel05}, Section 3.1]
  Let $R,C$ be permutations of $\{1,\ldots,n\}$ and $R_k,C_k$ defined
  as above.    If $|A^{(R_k,C_k)}|_{r_{k+1},c_{k+1}}$ is defined for
    $k=0\ldots n-1$ , then
    \[
    \Ddet(A) = \sign(R)\cdot \sign(C) \cdot 
    \pi( \PDdet_{R,C}(A)).
    \]
\end{fact}
In other words, the Dieudonn\'e determinant is essentially invariant
of the order of the sequence of submatrices specified in
\eqref{eq:predetexp}.

\enlargethispage{4pt}

\begin{thm}
Let $A\in\K^\nxn$ be invertible, with $\deg A_{ij}\leq d$.  Then 
$\deg\Ddet(A)\leq nd$.
\end{thm}
\begin{proof}
  We proceed by induction on $n$.  For $n=1$ this is clear.
  For $n=2$, the possible predeterminants are
  \begin{align*}
    \PDdet_{12,12}(A)=|A|_{11} A_{22} & = (A_{11}-A_{12}A_{22}^{-1}A_{21})A_{22},\\
    \PDdet_{12,21}(A)=|A|_{12} A_{21} & = (A_{12}-A_{11}A_{21}^{-1}A_{22})A_{21},\\
    \PDdet_{21,12}(A)=|A|_{21} A_{22} & = (A_{21}-A_{22}A_{12}^{-1}A_{11})A_{12},\\
    \PDdet_{21,21}(A)=|A|_{22} A_{11} & = (A_{22}-A_{21}A_{11}^{-1}A_{12})A_{11},
  \end{align*}
  at least one of which must be defined and non-zero, and all of which clearly have
  degree at most $2d$.

  Now assume the theorem is true for matrices of dimension less than
  $n$.  Choose $r_1,c_1\in\{1,\ldots,n\}$ such that $|A|_{r_1,c_1}$ is 
  non-zero and of minimal degree; that is
  $\deg |A|_{r_1,c_1} \leq \deg |A|_{k,\ell}$
  for all $k,\ell$ such that $|A|_{k,\ell}$ is defined and non-zero.
  The fact that $|A|_{r_1,c_1}\neq 0$ implies that
  $A^{(r_1,c_1)}$ is invertible, and we can continue this process
  recursively.  Thus, let $R=(r_1,\ldots,r_n)$ and
  $C=(c_1,\ldots,c_n)$ be permutations of $\{1,\ldots, n\}$ such that
  $|A^{(R_i,C_i)}|_{r_{i+1},c_{i+1}}\neq 0$ and $\deg
  |A^{(R_i,C_i)}|_{r_{i+1},c_{i+1}}$ is minimal over the degrees of
  non-zero, defined quasideterminants $|A^{(R_i,C_i)}|_{k,\ell}$, for $0\leq i<n$.
  Now
  \begin{align*}
   \PDdet_{R,C}(A)  & = |A|_{r_1,c_1} \cdot |A^{(r_1,c_1)}|_{r_2,c_2} \cdot
       \PDdet_{R,C}(A^{(R_2,C_2)}) \\
     & = \left(A_{{r_1,c_1}}-\sum_{k,\ell} A_{r_1k}|A^{(r_1,c_1)}|^{-1}_{\ell k}
       A_{\ell c_1}\right)\cdot |A^{(r_1,c_1)}|_{r_2,c_2}  \cdot
     \PDdet_{R,C}(A^{(R_2,C_2)})\\
    & = A_{r_1,c_1} \cdot |A^{(r_1,c_1)}|_{r_2,c_2}\cdot \PDdet_{R,C}(A^{(R_2,C_2)})\\
        & \hspace*{20pt}   - \sum_{k,\ell} A_{r_1k}|A^{(r_1,c_1)}|^{-1}_{\ell k}
     A_{\ell c_1} \cdot |A^{(r_1,c_1)}|_{r_2,c_2}  \cdot
     \PDdet_{R,C}(A^{(R_2,C_2)})\\
    & = A_{r_1,c_1} \cdot \PDdet_{R,C}(A^{(R_1,C_1)})\\
        & \hspace*{20pt}   - \sum_{k,\ell} A_{r_1k}|A^{(r_1,c_1)}|^{-1}_{\ell k}
     A_{\ell c_1} \cdot |A^{(r_1,c_1)}|_{r_2,c_2}  \cdot
     \PDdet_{R,C}(A^{(R_2,C_2)}),\\
  \end{align*}
where all sums are taken only over defined quasideterminants as above.
Thus
\[
  \deg\Ddet(A)  = \deg\PDdet_{R,C}(A) 
       \leq \max \left\{ d+(n-1)d, 2d+(n-2)d \right\} \leq nd,
\]
using the induction hypothesis and the assumption that 
$\deg |A^{(r_1,c_1)}|_{r_2,c_2}$ is chosen to be minimal.
\end{proof}

  \section{Degree bounds on matrices over $\FD$}
\label{sec:orebounds}


Some well-known properties of $\FD$ are worth recalling; see
\citep{Ore33} for the original theory or \citep{BroPet94} for an
algorithmic presentation.  Given $f,g\in\FD$, there is a degree
function (in $\D$) which satisfies the usual properties: $\degD
(fg)=\degD f+\degD g$ and $\degD(f+g)\leq \max\{\degD f,\degD g\}$.  We
set $\degD 0=-\infty$.

$\FD$ is a left (and right) principal ideal ring, which implies
the existence of a right (and left) division with remainder algorithm
such that there exists unique $q,r\in\FD$ such that $f=qg+r$ where
$\degD(r)<\degD(g)$.  This allows for a right (and left) Euclidean-like
algorithm which shows the existence of a greatest common right
divisor, $h=\gcrd(f,g)$, a polynomial of minimal degree (in $\D$) such
that $f=uh$ and $g=vh$ for $u,v\in\FD$.  The GCRD is unique up to a
left multiple in $\F\nonzero$, and there exist co-factors
$a,b\in\FD$ such that $af+bg=\gcrd(f,g)$.  There also exists a least
common left multiple $\lclm(f,g)$.  Analogously there exists a
greatest common left divisor, $\gcld(f,g)$, and least common right
multiple, $\lcrm(f,g)$, both of which are unique up to a right
multiple in $\F$. From \citep{Ore33} we also have that
\begin{equation}
\label{eq:oredeg}
\begin{aligned}
  \degD\lclm(f,g) &=\degD f + \degD g - \degD\gcrd(f,g),\\
  \degD\lcrm(f,g) & = \degD f + \degD g - \degD\gcld(f,g).
\end{aligned}
\end{equation}

It will be useful to work in the quotient skew field $\qFD$ of
$\FD$, and to extend the degree function $\degD$ appropriately.  We
first show that any element of $\qFD$ can be written as a
\emph{standard fraction} $fg^{-1}$, for $f,g\in\FD$ (and in
particular, since $\FD$ is non-commutative, we insist that $g^{-1}$ is
on the right).

\begin{fact}[\cite{Ore33}, Section 3]
  Every element of $\qFD$ can be written as a standard fraction.
\end{fact}

The notion of degree extends naturally to $\qFD$ as follows. 
\begin{defn}
  \label{def:deg}
  For $f,g\in\FD$, $g\neq 0$, the \emph{degree} $\degD
  (fg^{-1}) = \degD f-\degD g$.
\end{defn}



The proof of the next lemma is left to the reader.

\begin{lem}
  \label{lem:Oredeggood}
  For $f,g,u,v\in\FD$, with $g,v\neq 0$, we have the following:
  \begin{itemize}
  \item[(a)] if $fg^{-1}=uv^{-1}$ then $\degD(fg^{-1})=\degD(uv^{-1})$;
  \item[(b)] $\degD ( (fg^{-1}) \cdot (uv^{-1}))=\degD (fg^{-1})+\degD(uv^{-1})$;
  \item[(c)] $\degD(fg^{-1}+uv^{-1})\leq
    \max\{\degD(fg^{-1}),\degD(uv^{-1})\}$;
  \item[(d)]  $\degD ((fg^{-1})^{-1}) = -\degD (fg^{-1})$.
  \end{itemize}
\end{lem}

In summary, the degree function on $\qFD$ meets the requirement of a
degree function on a skew field as in Section \ref{sec:detbounds}, and
is once again, actually a valuation on $\qFD$.

\subsection{Determinantal degree and unimodularity}

We show unimodular matrices are precisely those with a Dieudonn\'e
determinant of degree zero.

\begin{lem}
  \label{lem:unimod}
  Let $W\in\FD^{2\times 2}$ be as in Fact \ref{fact:2x2uni}.  Then $\degD\Ddet W=0$.
\end{lem}
\begin{proof} 
  We may assume that $\gcrd(a,b)=g=1$, since the same matrix satisfies\linebreak
  $W(ag^{-1},bg^{-1})^T=(1,0)^T$.  Also assume both $a,b\neq 0$
  (otherwise the lemma is trivial).    Then
    \[
    \begin{pmatrix}
      u & v \\
      s & t
    \end{pmatrix}
    \begin{pmatrix}
      a & 0 \\
      b & 1
    \end{pmatrix}
    =
    \begin{pmatrix}
      1 & v\\
      0 & t
    \end{pmatrix},
    ~~~~\mbox{and}~~
    \Ddet(W) \cdot a \equiv t\bmod [\FD^*,\FD^*],
    \]
    so $\degD\Ddet W + \degD a = \degD t$.
    Since $\gcrd(a,b)=1$, from \eqref{eq:oredeg} we know $\degD a=\degD t$,
    so $\degD\Ddet W=0$.
\end{proof} 

Embedding the $2\times 2$ matrices into $n\times n$ identity matrices,
as in Lemma \ref{lem:nxnuni}, we obtain the following (the proof of
which is left to the reader).

\begin{cor}
  \label{cor:Eunimod}
  Let $E\in\FD^{n\times n}$ be as in Lemma \ref{lem:nxnuni}.  The
  $\degD\Ddet E = 0$.
\end{cor}

The characterization of unimodular matrices as those with Dieudonn\'e
determinant of degree zero follows by looking at the Hermite form of a
unimodular matrix.

\begin{thm}
   $U\in\FD^\nxn$ is unimodular if and only if
   $\degD\Ddet U=0$.
\end{thm}
\begin{proof}
  Suppose $U$ is unimodular.  The Hermite form of $U$ must be the
  identity: all the diagonal entries must be invertible in $\FD$ and
  the entries above the diagonal are reduced to $0$.  Thus, the
  unimodular multiplier to the Hermite form of $U$ will be the inverse
  $U$.

  Following the simple algorithm to compute the Hermite form in
  Theorem \ref{thm:sqherm}, we see it worked via a sequence of
  unimodular transforms, all of which were either permutations,
  off-diagonal reductions from Lemma \ref{lem:nxnuni}, or are of the form
  $E$ in Lemmas \ref{lem:nxnuni} and \ref{cor:Eunimod}.  The
  Dieudonn\'e determinants of permutations and reduction
  transformations are both equal to $1$, by the basic properties of Dieudonn\'e
  determinants discussed at the beginning of
  Section \ref{ssec:dieu}, and hence have degree $0$.  The
  Dieudonn\'e determinants of the transformations $E$ are of degree
  $0$ by Corollary \ref{cor:Eunimod}. The proof is now complete by the
  multiplicative property of Dieudonn\'e determinants, and the
  additive properties of their degrees.

  Assume conversely that $\degD\Ddet U=0$, and that $V\in\FD^\nxn$ is
  a unimodular matrix such that $VU=H$ is in Hermite form.  Then
  $\degD\Ddet V+\degD\Ddet U=\degD\Ddet H=0$.  But the only matrix in
  Hermite form with degree $0$ is the identity matrix.  Thus $V$ is
  the inverse of $U$, and $U$ must be unimodular.
\end{proof}
  
\subsection{Degree bounds on the Hermite form}

In this section we establish degree bounds on Hermite forms of
matrices over $\FD$ and their unimodular transformation matrices.

\begin{thm}
  \label{thm:Hdeg}
  Let $A\in\FD^\nxn$ have full rank and entries of degree at
  most $d$ and Hermite form $H\in\FD^\nxn$. Then 
  \begin{enumerate}
    \item[(a)] The sum of the degrees of the diagonal entries of $H$ has degree at
      most $nd$;
    \item[(b)] The sum of the degrees of the entries in any row of $H$ has
      degree at most $nd$.
    \end{enumerate}
\end{thm}
\begin{proof}
  Let $V\in\FD^\nxn$ be unimodular such that $A=VH$, whence
  $\Ddet(A)=\Ddet(V)\Ddet(H)$.  Therefore (a) follows from
  \[
  \degD\Ddet(A)=\degD\Ddet(H)=\sum_{1\leq i\leq n}\degD H_{ii}
  \leq nd.
  \]
  Point (b) follows from the fact that each entry above the diagonal
  in the Hermite form has, by definition, degree smaller than the
  degree of the diagonal entry below it.
\end{proof}

We now show that all entries in $H^{-1}$ have non-positive degrees.

\begin{lem}
  \label{lem:Hinvfrac}
  Let $H\in\FD^\nxn$ be of full rank and in Hermite form, and
  let $J=H^{-1}$.  Then $\degD J_{ij}\leq 0$ for $1\leq i,j\leq n$.
\end{lem}
\begin{proof}
  We consider the equation $JH=I$, and note that $J$, like $H$ is upper
  triangular. For each $r\in\{1,\ldots,n\}$ we show by induction on
  $c$ (for $r\leq c \leq n$) that $\degD J_{rc}\leq 0$.

  For the base case $c=r$, $J_{rr}H_{rr}=1$, so $\degD J_{rr}=-\degD
  H_{rr}\leq 0$.

  Assume now that $r<c$ and $\degD J_{r\ell}\leq 0$ for $r\leq \ell<c$.  We
  need to show that $\degD J_{rc}\leq 0$.  We know that
  \[
  \sum_{1\leq i\leq n} J_{r\ell}H_{\ell c}=\sum_{r\leq \ell\leq c} J_{r\ell} H_{\ell c} = 0.
  \]
  Since $\degD J_{r\ell}\leq 0$ for $r\leq\ell<c$ and $\degD H_{cc}>\degD
  H_{\ell c}$ for $r\leq\ell<c$, it must be the case that $\degD
  J_{rc}\leq 0$ as well.
\end{proof}

\begin{thm}
  \label{thm:HermDegV}
  Let $A\in \FD^\nxn$ be invertible (over $\qFD$), whose entries all
  have degree at most $d$ in $\D$.  Suppose $A$ has Hermite form $H\in
  \FD^\nxn$, with $UA=H$ and $UV=I$ for $U,V\in\FD^\nxn$.  Then $\degD
  V\leq d$ and $\degD U\leq (n-1)d$.
\end{thm}
\begin{proof}
  Note that $V=AH^{-1}$, and by Lemma \ref{lem:Hinvfrac} all entries
  in $H^{-1}$ have non-positive degree.  Thus $\degD V\leq \degD A$.  By
  Corollary \ref{cor:uniinvbd}, $\degD U\leq (n-1)d$.  
\end{proof}


  \section{Computing Hermite forms by linear systems over $\FD$}
\label{sec:compute}

In this section we present our polynomial-time algorithm to compute
the Hermite form of a matrix over $\FD$.  This generally follows the
``linear systems'' approach of \cite{KalKri87}, and more specifically
the refinements in \cite{Sto94} (for matrices over $\k[x]$ for a
field $\k$).  We will need the tools for $\FD$ we have
developed in the previous sections.  The method only directly
constructs the matrix $U$ such that $H=UA$. The Hermite form $H$ can
be found by performing the multiplication $UA$. 

The general approach is similar to that described in \cite{GieKim09},
with a primary difference that in that paper the technique of
\cite{KalKri87} was adapted. This new technique is considerably more
efficient (see below). As well, our earlier paper was
constrained to differential rings as the necessary
degree bounds were not available for all Ore polynomials.  

Assume that $A_{ij}=\sum_{0\leq k\leq d}A_{ijk}\D^k$ for
$A_{ijk}\in\F$. Let $\row(A,i)\in\FD^{1\times n}$ be the $i$th row of
$A$ and define
\[
\L(A) =\left\{\sum_{1\leq i\leq n} b_i \cdot \row(A,i): b_1,\ldots,b_n\in\FD
\right\},
\]
the left module of the row space of $A$.  The following lemma is
shown analogously to \cite[\S 4.3.1, Lemma 4]{Sto94}.
\begin{lem}
   \label{lem:latmem}
   Let $A\in\FD^\nxn$ be nonsingular, with Hermite form $H$.
   Let $h_i=\degD H_{ii}$ for $1\leq i\leq n$.  For
   $v=(0,\ldots,0,v_{\ell},\ldots,v_n)\in\FD^{1\times n}$, with $\degD
   v_\ell<h_\ell$,  then if $v\in\L(A)$ we have $v_\ell=0$, and if $v_\ell\neq 0$
   then $v\notin\L(A)$.
\end{lem}


The following theorem is analogous to \cite[\S4.3.1, Lemma 5]{Sto94},
with a different, slightly weaker degree bound.  

\begin{thm}
  \label{thm:herlinsys}
   Let $A\in\FD^\nxn$ have full rank, with $\degD A_{ij}\leq
   d$ for $1\leq i,j\leq n$.   Let $(d_1,\ldots,d_n)$ be a given
   vector of non-negative integers.  Let $T$ be an $n\times n$ matrix
   with $T_{ij}=\sum_{0\leq k\leq \varrho} t_{ijk}\D^k$ for unknowns
   $t_{ijk}$, where $\varrho\geq (n-1)d+\max_i \{d_i-h_i\}$.
  Consider the system of equations in $t_{ijk}$ with constraints:
   \begin{equation}
    \label{eq:TA}
   \begin{array}{llll}
    (TA)_{i,i,d_i} & = 1, & \mbox{for $1\leq i\leq n$,} & \emph{--- diagonal
      entries are monic;}\\
    (TA)_{i,i,k} & = 0, & \mbox{for $k>d_i$,} & \emph{--- diagonal entry
      in row $i$ has degree $d_i$;}\\
    (TA)_{i,j,k} & = 0, & \mbox{for $i\neq j$ and $k\geq d_j$} &
    \emph{--- off diagonal entries have lower degree}\\[-6pt]
    &&& \hspace*{14pt}\emph{than the diagonal entry in that column.}
   \end{array}
   \end{equation}

   By a \emph{solution for $T$} we mean an assignment of variables
   $t_{ijk}\gets \alpha_{ijk}\in\F$ for some $1\leq i,j\leq n$ and $0\leq
   k\leq \varrho$ such \eqref{eq:TA} holds.

   Let $h_1,\ldots,h_n\in\NN$ be the degrees of the diagonal entries
   of the Hermite form of $A$.  The following statements about the
   above system hold:
   \begin{enumerate}
   \item[(i)] If $d_i\geq h_i$ for $1\leq i\leq n$ then there exists a
     solution for $T$;
   \item[(ii)] If there exists a positive integer $\ell\leq n$ such
     that $d_i=h_i$ for $1\leq i<\ell$ and $d_\ell<h_\ell$ then there
     is no solution for $T$;
   \item[(iii)] If $d_i=h_i$ for $1\leq i\leq n$ then there is a
     unique solution for $T$ such that $G=TA$ is equal to the Hermite
     form of $A$ under that solution.
    \end{enumerate}
\end{thm}

\begin{proof}
  Let $H\in\FD^\nxn$ be the Hermite form of $A$ and let $U\in\FD^\nxn$
  be the unique unimodular matrix such that $UA=H$.

  To show (i), let $D=\diag(\D^{d_1-h_1},\ldots,$
  $\D^{d_n-h_n})\in\FD^\nxn$, and consider the equality $DUA=DH$.  Let
  $H^*\in\FD^\nxn$ be the Hermite form of $DH$ and $U^*\in\FD^\nxn$
  the unimodular matrix such that $U^*DUA=H^*$.  We construct $H^*$ 
  from $DH$ simply by reducing the entries above the diagonal (since
  it is already upper triangular).  Thus $U^*$ is upper triangular,
  with ones on the diagonal, and $\degD U^*_{ij}< (d_i-h_i)-(d_j-h_j)$
  for $i<j$.  We claim $T=U^*DU$ is a solution to \eqref{eq:TA} .
  First note that the particular choice of $D$, together with the
  definition of $H^*$ ensure that the constraints of \eqref{eq:TA} are
  met.  Furthermore, entries in the $i$th row of $U^*D$ have degree at
  most $d_i-h_i$.  By Theorem \ref{thm:HermDegV}, $\degD U\leq (n-1)d$,
  hence $\degD T \leq (n-1)d+\max_i\{d_i-h_i\}\leq \varrho$.


  To prove (ii), suppose by contradiction that there exists a
  nonnegative integer $\ell\leq n$ and a solution $T$ such that
  $\degD((TA)_{ii})=d_i$ for $1\leq i<\ell$ and
  $\degD((TA)_{\ell\ell})<h_\ell$.  Note that
  $\row(TA,\ell)=((TA)_{\ell,1},\ldots,(TA)_{\ell,n})$ is in $\L(A)$.
  First, if $\ell=1$, then $\degD (TA)_{\ell,1}<h_1$, so by Lemma
  \ref{lem:latmem}, $(TA)_{1,1}=0$, which is impossible since
  \eqref{eq:TA} ensures this entry is monic.  Now assume $\ell>1$.
  Then $\degD (TA)_{\ell,1}<h_1$ (to satisfy \eqref{eq:TA}), and hence
  by Lemma \ref{lem:latmem}, so $(TA)_{\ell,1}=0$.  A simple induction
  shows that $(TA)_{\ell,j}=0$ for $1\leq j<\ell$.  Now consider
  $(TA)_{\ell,\ell}$, which has degree $d_\ell<h_\ell$ by our
  assumption.  Again by Lemma \ref{lem:latmem} $(TA)_{\ell,\ell}=0$,
  which \eqref{eq:TA} ensures is monic, a contradiction.

  If the conditions of (iii) hold, then by (i) there exists at least
  one solution for $T$.  We can use an inductive proof similar to that
  used in our proof of (ii) to show that elements below the diagonal
  in $TA$ are zero (i.e., that $(TA)_{ij}=0$ for $i>j$).  By the
  uniqueness of the Hermite form we must have $TA=H$.
\end{proof}

This theorem allows us to work with a partial order on the degree sequences.
For any $(h_1,\ldots,h_n),(d_1,\ldots,d_n)\in\ZZ^n$, we say that
$(h_1,\ldots,h_n)\preceq (d_1,\ldots,d_n)$ if and only if $h_i\leq d_i$ for all $1\leq
i\leq n$ (and similarly define $\precneqq$ for strict precedence).  Thus, \eqref{eq:TA}
has a solution if and only if $(h_1,\ldots,h_n)\preceq (d_1,\ldots,d_n)$ and
this is unique if and only if $(h_1,\ldots,h_n)=(d_1,\ldots,h_n)$.

We now embed the system \eqref{eq:TA} into a system of linear
equations over $\F$, with no Ore component.  We
embed $\FD$ into vectors over $\F$ via $\tau_\ell:
\FD \to \F^{\ell+1}$, with
\[
\tau_\ell(u_0+u_1\D + u_2 \D^2 + \cdots + u_\ell \D^{\ell-1})  
= 
(u_0,\ldots,u_\ell) \in\F^{\ell+1}.
\]
For $g\in\FD$ of degree $d$, $u\in\FD$ of degree at most $m$, and
assuming $\ell\geq m+d$, the equation $ug=f$ can be realized by a
matrix equation over $\F$:
\[
(u_0,\ldots,u_m) 
\kern-10pt
\underbrace{\left(
  \begin{array}{ccc}
  \tau_\ell(g) \\ \hline
  \tau_\ell(\D g) \\ \hline
  \vdots\\ \hline
  \tau_\ell(\D^m g)
\end{array}
\right)}_{\displaystyle \mu_m^\ell(g)\in\F^{(m+1)\times(\ell+1)}}
\kern-10pt =
(f_0,\ldots,f_\ell)
~~~\iff~~~
\tau_m(u)\,\mu_m^\ell(g)=\tau_\ell(f).
\]

Fixing $d_1,\ldots,d_n\in\NN$ as in Theorem \ref{thm:herlinsys}, and
setting $\varrho\geq (n-1)d+\max_i\{d_i-h_i\}$, we
can then study \eqref{eq:TA}, as realized as (a subset of) the linear
equations in the matrix equation over $\F$:
\begin{equation}
\label{eq:TA=G}
\underbrace{
  \begin{pmatrix}
    \tau_{\varrho}(T_{11}) & \cdots & \tau_{\varrho}(T_{1n}) \\
    \vdots & & \vdots \\
    \tau_{\varrho}(T_{n1}) & \cdots & \tau_{\varrho}(T_{nn}) 
  \end{pmatrix}
}_{\Tlin \in \F^{n\times (\varrho+1)n}}
\underbrace{
  \begin{pmatrix}
    \mu_{\varrho}^{\varrho+d}(A_{11}) & \cdots & \mu_{\varrho}^{\varrho+d}(A_{1n}) \\
    \vdots &  & \vdots \\
    \mu_{\varrho}^{\varrho+d}(A_{n1}) & \cdots & \mu_{\varrho}^{\varrho+d}(A_{nn})
  \end{pmatrix}
}_{\Alin\in\F^{n(\varrho+1)\times (\varrho+d+1)n}}
=
\underbrace{
  \begin{pmatrix}
    \tau_{\varrho+d}(G_{11}) & \cdots & \tau_{\varrho+d}(G_{1n}) \\
    \vdots & & \vdots \\
    \tau_{\varrho+d}(G_{n1}) & \cdots & \tau_{\varrho+d}(G_{nn})
  \end{pmatrix}.
}_{\Glin\in\F^{n\times (\varrho+d+1) n}}
\end{equation}
This set of equations is a superset of the equation \eqref{eq:TA}.
Some entries in $\Glin$ are unknown, in particular those corresponding
to coefficients of degrees (in $\D$) strictly less than the degree of
the diagonal below it.  However, these entries in $\Glin$ are not
mentioned or involved in Theorem \ref{thm:herlinsys}, and we can
remove these columns from $\Glin$.  Similarly, since they impose no
constraint on \eqref{eq:TA}, we can remove the corresponding columns
of $\Alin$.  By Theorem \ref{thm:herlinsys}, if we know
$d_1,\ldots,d_n$, the remaining equations will have a unique solution,
from which we completely determine $\Tlin$.

\enlargethispage{10pt}

\begin{exmp}
\label{ex:hermcomp}
Consider the following matrix in $\QQ(z)[\D;\diff]$
(the differential polynomials over $\QQ(z)$):
\[
A=
\begin{pmatrix}
(z+1)+\D & z+ z\D & \D \\ 
(z^2+z)+z\D  &z+1 &2\D \\ 
(-z-z^2)-z\D & z\D &z\D 
\end{pmatrix}\in\QQ(z)[\D;\diff]^{3\times 3}.
\]
Assume for this example that we know the degrees of the entries in the
Hermite form are $(d_1,d_2,d_3)=(1,0,2)$.  Then $n=3$, and we can set
$\varrho=2$, and have

{\small
\[
\overbrace{
\left(
\begin{array}{ccc|ccc|ccc}
t_{110} & t_{111} & t_{112} & t_{120} & t_{121} & t_{122} & t_{130} &
t_{131} & t_{132} \\\hline
t_{210} & t_{211} & t_{212} & t_{220} & t_{221} & t_{222} & t_{230} &
t_{231} & t_{232} \\\hline
t_{210} & t_{211} & t_{212} & t_{220} & t_{221} & t_{222} & t_{230} & t_{231} & t_{232} 
\end{array}
\right)}^{\mbox{\normalsize $\Tlin$}}
\overbrace{
\left(
\begin{array}{cccc|cccc|cccc}
z+1 & 1 & 0 & 0 & z & z & 0 & 0 & 0 & 1 & 0 & 0 \\
1 & z+1 & 1 & 0 & 1 & z+1 & z & 0 & 0 & 0 & 1 & 0 \\
0 & 2 & z+1 & 1 & 0 & 2 & z+2 & z & 0 & 0 & 0 & 1 \\ \hline
z^2+z & z & 0 & 0 & z+1 & 0 & 0 & 0 & 0 & 2 & 0 & 0 \\
2z+1 & z^2+z+1 & z & 0 & 1 & z+1 & 0 & 0 & 0 & 0 & 2 & 0 \\
2 & 4z+2 & z^2+z+2 & z & 0 & 2 & z+1 & 0 & 0 & 0 & 0 & 2 \\ \hline
-z^2-z & -z & 0 & 0 & 0 & z & 0 & 0 & 0 & z & 0 & 0 \\
-2z-1 & -z^2-z-1 & -z & 0 & 0 & 1 & z & 0 & 0 & 1 & z & 0 \\
-2 & -4z-2 & -z^2-z-2 & -z & 0 & 0 & 2 & z & 0 & 0 & 2 & z \\
\end{array}
\right)}
^{\mbox{\normalsize $\Alin$}}
\]
\[
=
\overbrace{
\left(
\begin{array}{cccc|cccc|cccc}
G_{110} & 1 & 0 & 0 & 0 & 0 & 0 & 0 & G_{130} & G_{131} & 0 & 0 \\
G_{210} &  0 & 0 & 0 & 1 & 0 & 0 & 0 & G_{230} & G_{231} & 0 & 0 \\
G_{310} &  0 & 0 & 0 & 0 & 0 & 0 & 0 & G_{330} & G_{331} & 1 & 0
\end{array}
\right)}^{\mbox{\normalsize $\Glin$}}.
\]
} 

As noted above, $\Glin$ still has some indeterminates, from columns
which specify coefficients of the entries of $G$ which are of degree
strictly less than the maximum in the corresponding column of $G$.
These entries are not mentioned in \eqref{eq:TA}, and we
remove them to form $\Glinred$.  The corresponding columns of $A$ are
similarly not involved in \eqref{eq:TA}, and are removed to form
$\Alinred$.  We obtain now obtain a reduced system of equations which
corresponds precisely to \eqref{eq:TA}:
{\small
\[
\overbrace{
\left(
\begin{array}{ccc|ccc|ccc}
t_{110} & t_{111} & t_{112} & t_{120} & t_{121} & t_{122} & t_{130} &
t_{131} & t_{132} \\\hline
t_{210} & t_{211} & t_{212} & t_{220} & t_{221} & t_{222} & t_{230} &
t_{231} & t_{232} \\\hline
t_{210} & t_{211} & t_{212} & t_{220} & t_{221} & t_{222} & t_{230} & t_{231} & t_{232} 
\end{array}
\right)}^{\mbox{\normalsize $\Tlin$}}
\overbrace{
\left(
\begin{array}{ccc|cccc|cc}
1 & 0 & 0 & z & z & 0 & 0 & 0 & 0 \\
z+1 & 1 & 0 & 1 & z+1 & z & 0 & 1 & 0 \\
2 & z+1 & 1 & 0 & 2 & z+2 & z & 0 & 1 \\ \hline
z & 0 & 0 & z+1 & 0 & 0 & 0 & 0 & 0 \\
z^2+z+1 & z & 0 & 1 & z+1 & 0 & 0 & 2 & 0 \\
4z+2 & z^2+z+2 & z & 0 & 2 & z+1 & 0 & 0 & 2 \\ \hline
-z & 0 & 0 & 0 & z & 0 & 0 & 0 & 0 \\
-z^2-z-1 & -z & 0 & 0 & 1 & z & 0 & z & 0 \\
-4z-2 & -z^2-z-2 & -z & 0 & 0 & 2 & z & 2 & z \\
\end{array}
\right)}^{\mbox{\normalsize $\Alinred$}}
\]
\[
=
\overbrace{
\left(
\begin{array}{ccc|cccc|cccc}
1 & 0 & 0 & 0 & 0 & 0 & 0 & 0 & 0 \\
0 & 0 & 0 & 1 & 0 & 0 & 0 & 0 & 0 \\
0 & 0 & 0 & 0 & 0 & 0 & 0 & 1 & 0
\end{array}
\right)}^{\mbox{\normalsize $\Glinred$}}
\in\F^{n(\varrho+1)\times n(\varrho+1)}.
\]}

By Theorem \ref{thm:herlinsys}, since we have ``guessed'' the degree
sequence of the diagonal entries $(d_1,d_2,d_3)=(1,0,2)$ correctly,
the system has a unique solution:
\[
\Tlin
=
\begin{pmatrix}
{\frac {z+1}{2\,z+1}}&0&0&-{\frac {z }{2\,z+1}}&0&0&-{\frac {z+1}{2\,z+1}}&0&0\\ 
-{\frac{z}{2\,z+1}}&0&0&{\frac {z+1}{2\,z+1}}&0&0&{\frac  {z}{2\,z+1}}&0&0 \\ 
-{\frac {2\,{z}^{2}+3\,z+2}{ \left( {z}^{2}+z+2 \right)  \left( 2\,z+1 \right) }}&-{\frac {z}{{z}^{2}+z+2}}&0&{\frac 
{2\,{z}^{2}+z-1}{ \left( {z}^{2}+z+2 \right)  \left( 2\,z+1 \right) }}
&{\frac {z+1}{{z}^{2}+z+2}}&0&{\frac {2\,{z}^{3}-{z}^{2}-2\,z-1}{z
 \left( {z}^{2}+z+2 \right)  \left( 2\,z+1 \right) }}&{\frac {z}{{z}^{
2}+z+2}}&0
\end{pmatrix}
\]
which corresponds to
\[
T=
\begin{pmatrix}
{\frac {z+1}{2\,z+1}} &   -{\frac {z }{2\,z+1}}&-{\frac {z+1}{2\,z+1}}\\ 
-{\frac{z}{2\,z+1}}&{\frac {z+1}{2\,z+1}}&{\frac {z}{2\,z+1}} \\ 
-{\frac {2\,{z}^{2}+3\,z+2}{ \left( {z}^{2}+z+2 \right)  \left(
      2\,z+1 \right) }} -{\frac {z}{{z}^{2}+z+2}} \D\  & {\frac 
{2\,{z}^{2}+z-1}{ \left( {z}^{2}+z+2 \right)  \left( 2\,z+1 \right)
}} + {\frac {z+1}{{z}^{2}+z+2}}\D\  &{\frac {2\,{z}^{3}-{z}^{2}-2\,z-1}{z
 \left( {z}^{2}+z+2 \right)  \left( 2\,z+1 \right) }}) + {\frac {z}{{z}^{
2}+z+2}}\D\  
\end{pmatrix}
\in\QQ(z)[\D;\diff]^{3\times 3}
\]
giving
\[
H=TA=
\begin{pmatrix}
(z+1)+\D & 0 & -\frac{z^2+2z-1}{2z+1}\D \\
0 & 1 & \frac{z^2+z+2}{2z+1}\D\\
0 & 0 & \frac{2z^3+3z^2-2z-5}{(z^2+z+2)(2z+1)}\D+\D^2
\end{pmatrix}
\in\QQ(z)[\D;\diff]^{3\times 3}
\]
in Hermite form.
\end{exmp}

We can now state our algorithm for computing the Hermite form given the
degrees of the diagonal elements.

\begin{alg}[alg:HermDegs]{HermiteFormGivenDegrees}
  \REQUIRE $A\in\FD^\nxn$ of full rank, with (unknown) Hermite
  form $H$ with diagonal degrees
  $(h_1,\ldots,h_n)\in\NN^n$; 
  \REQUIRE $(d_1,\ldots,d_n)\in\NN^n$, the
  proposed degrees of the diagonal entries of $H$ 
  \ENSURE $H\in\FD^\nxn$ if
  $(d_1,\ldots,d_n)=(h_1,\ldots,h_n)$, or a message that
  $(d_1,\ldots,d_n)$ is lexicographically smaller or larger than
  $(h_1,\ldots,h_n)$;
  
  \STATE Let $\varrho=(n-1)d+\max_i d_i$
  \STATE Form the matrix equation $\Tlin\Alin=\Glin$ as in
  \eqref{eq:TA=G}
  \STATE Remove all columns from $\Glin$ containing an indeterminate,
  and corresponding columns from $\Alin$, to form the ``reduced'' linear
  system $\Tlin\Alinred=\Glinred$, where $\Alinred$ and $\Glinred$ are
  now matrices over $\F$
  \IF{$\rank\Alinred < (n+1)\varrho$}
  \RETURN ``$(h_1,\ldots,h_n) \precneqq (d_1,\ldots,d_n)$"  \qquad // System is underconstrained
  \ENDIF
  \IF{$\Tlin\Alinred=\Glinred$ has no solution}
  \RETURN ``$(h_1,\ldots,h_n)\npreceq (d_1,\ldots,d_n)$"     \qquad // System is inconsistent
  \ENDIF
  \STATE Solve the system $\Tlin\Alinred=\Glinred$ for $\Tlin$ 
  \STATE Construct $T\in\FD^\nxn$ from $\Tlin$
  \RETURN $H=TA$ and $U=T$
\end{alg}

From Theorem \ref{thm:Hdeg} we know that each entry in the Hermite form
of $A\in\FD^{n\times n}$, with $\degD A_{ij}\leq d$ for $1\leq
i,j\leq n$, has degree at most $nd$.  If the diagonal entries of $A$ have degrees $(h_1,\ldots,h_n)$,
then we know that 
\[
(0,\ldots,0) \preceq (h_1,\ldots,h_n) \preceq (nd,nd,\ldots,nd).
\]
Algorithm \ref{alg:HermDegs} detects whether our choice of degree sequence is
equal to, larger than, or not larger than or equal to the actual one.  Thus, a
simple component-wise binary search allows us to find the actual degree sequence
$(h_1,\ldots,h_n)$.  That is, start by finding for the $h_1$ by executing
\ref{alg:HermDegs} with degree sequence $(d_1,nd,\ldots,nd)$ for different
values of $d_1$.  This will require $O(\log(nd))$ attempts.  Then search for
$h_2$ using degree sequence $O(h_1,d_2,nd,\ldots,nd)$ for different values of
$d_2$, etc.  It will require at most $O(n\log(nd))$ attempts to find the entire
correct degree sequence $(h_1,\ldots,h_n)$.

\begin{lem}
  \label{lem:HermBS}
  Given $A\in\FD^\nxn$ of full rank, where each entry has degree (in
  $\D$) less than $d$, we can compute the Hermite form $H\in\FD^\nxn$
  of $A$, and $U\in\FD^\nxn$ such that $UA=H$. The algorithm requires
  us to call Algorithm \ref{alg:HermDegs} $O(n\log(nd))$ times, with input $A$
  and varying degree sequences.
\end{lem}

For a first, general analysis of the complexity we will assume that
operations in $\F$ have unit cost (and hence no coefficient growth is
accounted for).  To perform the rank test in Step 4, the inconsistency
test in Step 6, and the equation solution in Step 8, we can simply do
an LU decomposition of $\Alinred$ using Gaussian elimination.
$\Alinred$ has size $n(\varrho+1)\times m$, where $n(\varrho+1)\leq
m\leq n(\varrho+d+1)$, i.e., $O(n^2d) \times O(n^2d)$.  Gaussian
elimination, which computes an $LU$-decomposition or more generally a
Bruhat normal form (see Section \ref{ssec:dieu} or \citep[Chapter
19]{Dra83}) is effective over any skew field, and on a $p\times q$
matrix requires $O(p^2q)$ operations, and hence in our case can be
accomplished with $O(n^6d^3)$ operations in $\F$.  Combining this with Lemma
\ref{lem:HermBS} we obtain the following.

\begin{thm}
  \label{thm:sqhermFcost}
  Let $A\in\FD^\nxn$ have full rank with entries of degree (in
  $\D$) less than $d$.  We can compute the Hermite form $H\in\FD^\nxn$
  of $A$, and $U\in\FD^\nxn$ such that $UA=H$. The algorithm requires
  $O(n^7d^3\log(nd))$ operations in $\F$.
\end{thm}

We next analyze our algorithm for computing the Hermite form of a
matrix $A\in\k(z)[\D;\sigma,\delta]^\nxn$ over the field $\F=\k(z)$,
where $\k$ is a field and $z$ an indeterminate.  Without loss of
generality $A\in\k[z][\D;\sigma,\delta]^\nxn$ by clearing denominators
(which is a left-unimodular operation), but note that the Hermite form
may still be in $\k(z)[\D;\sigma,\delta]$ (see Example \ref{ex:hermcomp}).
We will also assume for convenience that $\sigma(z)\in\k[z]$ and
$\deg_z\delta(z)\leq 1$.  Thus $\D z=\sigma(z)\D+\delta(z)\in\k[z][\D]$
and the degree in $z$ and $\D$ remains unchanged.  A more general
analysis could follow similarly.

We assume that multiplying two polynomials in $\k[z]$ of degree at
most $m$ can be accomplished with $O(\M(m))$ operations in $\k$:
$\M(m)=m^2$ using standard arithmetic or $\M(m)=m\log m \log\log m$
using fast arithmetic \citep{CanKal91}.  We similarly assume that two
integers with $\ell$ bits can be multiplied with $O(\M(l))$ bit
operations.  Finally, when we talk of the \emph{degree} of a rational
function in $\k(z)$ we mean the maximum degree of the numerator and
denominator, assuming they are co-prime.  This gives a reasonable
indication of representation size.

\begin{thm}
  \label{thm:sqhermcost}
  Let $A\in\k[z][\D;\sigma,\delta]^\nxn$ have full rank with
  entries of degree at most $d$ in $\D$, and of degree at most $e$ in
  $z$.  Let $H\in\k(z)[\D;\sigma,\delta]^\nxn$ be the Hermite form of
  $A$ and $U\in\k(z)[\D;\sigma,\delta]^\nxn$ such that $UA=H$.  
  \begin{enumerate}
    \item[(a)] $\deg_z H_{ij} \in O(n^2de)$ and $\deg_z U_{ij}\in
      O(n^2de)$ for $1\leq i,j\leq n$.
    \item[(b)] We can compute $H$ and $U$ deterministically with 
      $O(n^7d^3\log(nd)\cdot\M(n^2de))$ or $\softO(n^9d^4e)$
      operations in $\k$.
    \item[(c)] Assume $\k$ has at least $4n^2de$ elements.  We can
      compute the Hermite form $H$ and $U$ with an expected number of
      $O(n^7d^3\log(nd)+n^7d^3e)$ of operations in $\k$ using
      standard polynomial arithmetic.  This algorithm is
      probabilistic of the Las Vegas type; it never returns an
      incorrect answer.
    \end{enumerate}
\end{thm}
\begin{proof}
  To show (a), recall that the matrix $\Alinred$ is of size
  $O(n^2d)\times O(n^2d)$ and degree $O(e)$.  Using Hadamard's bound and Cramer's rule,
  the numerators and denominators in $\Tlin$ thus have degree at most
  $O(n^2de)$ in $z$.  $H=UA$ has the same degree bound in $z$.

  To prove (b) we solve the system of equations \eqref{eq:TA} as in
  Theorem \ref{thm:sqhermFcost} but now taking into account
  coefficient growth.  Since we have an explicit bound on the degree
  in $z$ of numerators and denominators of the solution, we can
  compute modulo an irreducible polynomial $\Gamma\in\k[z]$ more than
  twice this degree and recover the solution over $\k(z)$ by rational
  recovery.  Each operation in $\k(z)\in\k[z]/(\Gamma)$ will thus take
  $\softO(\M(n^2de))$ operations in $\k$.  The stated total cost
  follows from the cost in Theorem \ref{thm:sqhermFcost} multiplied by
  this operation cost.

  To show (c), we note that the tests for rank deficiency in Step 4,
  and inconsistency in Step 6, can be done by considering the equation
  $\Tlin\Alinred=\Glinred\bmod (z-\alpha)$ for a randomly chosen
  $\alpha$ from a subset of $\k$ of size at least $4n^2de$.  This
  follows because the largest invariant factor $w\in\k[z]$ of
  $\Alinred$ has degree at most $n^2de$ by Hadamard's bound (see part
  (a)), and the rank modulo $(z-\alpha)$ changes only if $\alpha$ is a
  root of $w$.  By the Schwartz-Zippel Lemma \citep{Sch80} this
  happens with probability at most $1/4$ for each choice of $\alpha$
  (and this probability of error can be made exponentially smaller by
  repeating with different random choices).  Thus, these tests require
  only $O(n^6d^3)$ operations in $\k$ to perform, correctly with high
  probability.  During the binary search for the degree sequence we
  only perform these cheaper tests, requiring a total of
  $O(n^7d^3\log(nd))$ operations in $\k$ before finding the correct
  degree sequence.

  Once we have found the correct degree sequence, we employ Dixon's
  \citeyearpar{Dix82} algorithm to solve the linear system over
  $\k(z)$ (this is the fastest known algorithm using standard matrix
  arithmetic, and is very effective in practice; one could also employ
  the asymptotically faster method of \cite{Sto03} with sub-cubic
  matrix arithmetic).  This lifts the
  solution to the system modulo $(z-\alpha)^i$ for
  $i=1,\ldots,2n^2de$, where $\alpha$ is a non-root of the (unknown)
  largest invariant factor of $A$ (i.e., is such that $\rank A=\rank
  A\bmod (z-\alpha)$).  Computing the solution modulo
  $(z-\alpha)^{2n^2de}$ is sufficient to recover the solution in
  $\k(z)$ using rational function reconstruction, since both the
  numerator and denominator have degree less than $n^2de$ by part (a);
  see \citep{GatGer03}, Section 5.7.  A random choice of $\alpha$ from
  a subset of $\k$ of size $4n^2de$ is sufficient to obtain a non-zero
  of the largest invariant factor (and hence not change the dimension
  of the solution space) with probability at least $1/4$ by the
  Schwartz-Zippel Lemma.  In the first step of Dixon's algorithm, we
  compute the LU-decomposition of $A\bmod (z-\alpha)$ using
  $O(n^6d^3)$ operations in $\k$.  We then lift the solution to
  $\Tlin\Alinred\equiv\Glinred\bmod (z-\alpha)^i$ for
  $i=0,\ldots,2n^2de$.  Each lifting step requires $O(n^5d^2)$
  operations in $\k$, yielding a total cost of $O(n^7d^3e)$.
\end{proof}

For comparison, the cost of the algorithm in \citep{GieKim09}, for the
case of matrices over $\k(z)[\D;\diff]$, required
$\softO(n^{10}d^4e)$ operations in $\k$.

Finally, we consider coefficient growth in $\QQ$ of Ore polynomial
rings over $\QQ(z)$.  For the computation, once we have constructed
the matrix $\Alinred$, we can bound the coefficient-sizes in $\Tlin$
directly using Hadamard-type bounds.  We can then employ a Chinese
remainder scheme to find the Hermite form using the above algorithm
(or any other method, for that matter).  For example, we could simply
choose a single prime $p$ with twice as many bits as the largest
numerator or denominator in the solution to \eqref{eq:TA=G} and then
compute modulo that prime, in $\ZZ_p[z]$; the rational coefficients of
$H$ can be recovered by integer rational reconstruction from their
images in $\ZZ_p$ \cite[\S 5.7]{GatGer03}.  

However, for the purposes of analysis, it is interesting to see how
big these coefficients can grow.  We consider matrices in
$A\in\ZZ[z][\D;\sigma,\delta]^\nxn$ without loss of generality.  For
convenience in this analysis (though not in complete generality), we assume
that $\deg_z\delta(z)\leq 1$ and $\sigma(z)\in\ZZ[z]$, so $\D
z=\sigma(z)\D+\delta(z)\in\ZZ[z]$.

For a polynomial $a=a_0+a_1z+\cdots+a_mz^m\in\ZZ[z]$, let
$\inorm{a}=\max_i|a_i|$.  For
$f=f_0(z)+f_1(z)\D+\cdots+f_r(z)\D^r\in\ZZ[z][\D;\sigma,\delta]$, let
$\inorm{f}=\max_i\inorm{f_i}$.  Define $\inorm{A}=\max_{ij}
\inorm{A_{ij}}$.  In equation \eqref{eq:TA=G}, the entries in $\Alinred$
have size at most
\begin{equation}
\label{eq:Abd}
\rnorm{A} = \max_{ij} \max_\ell \left\{ \inorm{A_{ij}}, \inorm{\D A_{ij}},
    \ldots,\inorm{\D^\varrho A_{ij}}\right\}\in\ZZ.
\end{equation}

\begin{thm}
  \label{thm:bitbd}
  Let $A\in\ZZ[z][\D;\sigma,\delta]^\nxn$ be of full rank and such
  that $\degD(A)=d$, $\deg_z(A)\leq e$ and $\rnorm{A}\leq\beta$.
  Then the Hermite form $H\in\QQ(z)[\D;\sigma,\delta]^\nxn$ and
  $U\in\QQ(z)[\D;\sigma,\delta]^\nxn$ such that $UA=H$ satisfy
  \[
  \log\inorm{H},\ \log\inorm{U} \in O(n^2d\,(\log e+\log\beta+\log n +
  \log d)).
  \]
\end{thm}
\begin{proof}
  Entries in $\Alinred$ are polynomials in $\ZZ[z]$ of degree at most $e$
  and coefficient size at most $\beta$.  Every minor of $\Alinred$, and
  hence each entry in the solution $\Tlin$, is bounded by Hadamard's
  bound, which in this case is
  \[
  \left ((1+e) \beta  (n^2d)\right)^{O(n^2d)}
  \]
  (see \cite{Gie93} Theorem 1.5 for height bounds on polynomial products).
\end{proof}

By performing all computations modulo an appropriately large prime (as
discussed above), we immediately get the following.
\begin{cor}
  Let $A\in\ZZ[z][\D;\sigma,\delta]^\nxn$ have full rank with
  entries of degree at most $d$ in $\D$, of degree at most $e$ in $z$,
  and $\rnorm{A} \leq \beta$ (where $\varrho=O(n^2d)$).  Let
  $H\in\QQ(z)[\D;\sigma,\delta]^\nxn$ be the Hermite form of $A$ and
  $U\in\QQ(z)[\D;\sigma,\delta]^\nxn$ such that $UA=H$.

  We can compute the Hermite form $H\in\QQ(z)[\D;\sigma,\delta]^\nxn$
  of $A$, and $U\in\QQ(z)[\D;\sigma,\delta]^\nxn$ such that $UA=H$,
  using an algorithm that requires an expected number
  $O((n^7d^3\log(nd)+n^7d^3e)\cdot\M(n^2d(\log
  e+\log\beta+\log n+\log d)))$,
  or $\softO(n^9d^4e\log\beta)$ bit operations.  This
  algorithm is probabilistic of the Las Vegas type (never returning an
  incorrect answer).
\end{cor}

The following corollary summarizes this growth explicitly over two
common rings, the differential polynomials $\QQ(z)[\D;\diff]$, and the
shift polynomial $\QQ(z)[\D;\shift]$.

\begin{cor}
  Let $A\in\ZZ[z][\D;\sigma,\delta]^\nxn$ be of full rank and such
  that $\degD(A)=d$, $\deg_z(A)\leq e$,
  $H\in\QQ(z)[\D;\sigma,\delta]^\nxn$ the Hermite form of $A$, and
  $U\in\QQ(z)[\D;\sigma,\delta]^\nxn$ such that $UA=H$.  For both the
  differential polynomials $\QQ(z)[\D;\diff]$ (where $\sigma(z)=z$, $\delta(z)=1$) and the
  shift polynomials $\QQ(t)[\D;\shift]$ (where $\sigma(z)=z+1$, $\delta(z)=0$), we have
  \[
  \log\inorm{U},\, \log\inorm{H} \in \softO(n^2d\,(e+\log\inorm{A})).
  \]
\end{cor}
\begin{proof}
  To show this for differential polynomials, we note that for
  $a=\sum_{0\leq i\leq d} a_i(z)\D^i\in\ZZ[z][\D;\diff]$,
  \[
  \D^\ell a = \sum_{0\leq j\leq\ell} \binom{\ell}{j} \sum_{0\leq i\leq
    d} a_i(z)^{(j)}\D^{\ell-j},
  \]
  where $a_i(z)^{(j)}$ is the $j$th derivative of $a_i(z)$.  Since
  only the first $e$ derivatives of any $a_i$ are non-zero
  \[
  \inorm{\D^\ell a}\leq  \ell^e\cdot\inorm{a}\cdot e!
  \]
  and hence $\log\rnorm{A}\in O(\log\inorm{A}+e\log(n^2d))$ for
  $\varrho=O(n^2d)$.  The result follows by Theorem \ref{thm:bitbd}.

  To show this for shift polynomials we note that for $a=\sum_{0\leq
    i\leq d} a_i(z)\D^i\in\ZZ[\D,\shift]$,
  \[
  \D^\ell a = \sum_{0\leq i\leq d} a_i(z+\ell)\D^i,
  \]
  so
  \[
  \inorm{\D^\ell a}\leq \inorm{a}2^{e/2}\ell^e,
  \]
  and hence $\log\rnorm{A}\in O(\log\inorm{A}+e\log(n^2d))$ for
  $\varrho=n^2d$.
  Again, the result follows from Theorem \ref{thm:bitbd}.
\end{proof}


  \section{Rectangular and rank deficient matrices}
\label{sec:oddshapes}

To this point we have assume that our matrices $A\in\FD^\nxn$ were
both square and of full rank.  In this section we relax both these
conditions to show a polynomial-time algorithm for the Hermite form in
all cases.

\subsection{Matrices with non-full rank}

\label{ssec:nonfullrank}

Suppose now that $A\in\FD^{m\times n}$ has rank $r\leq m$.  We
first show how to compute a unimodular matrix $P\in\FD^\nxn$ such that
$PA$ has precisely $r$ non-zero rows.  Since $P$ is unimodular, the
left $\FD$-modules generated by the rows of $A$ and the rows of $PA$
are equal.

We employ the row reduction algorithm developed by
\cite{BeckermannChengLabahn:2006}, and discussed in
\citep{DaviesChengLabahn:2008}.  Let $b=m\cdot \degD A$ and
$Q=\diag(\partial^b,\ldots,\partial^b)\in\FD^\nxn$, and form the
matrix
\[
B = \begin{pmatrix}
    AQ\\
    -I_n 
    \end{pmatrix} \in\FD^{(m+n)\times n}.
\]
We compute a basis for the left nullspace basis of $B$ in Popov
form using the algorithm in \cite{BeckermannChengLabahn:2006}, and
suppose it has form
\[
\begin{pmatrix}
P &|& R
\end{pmatrix}
\in\FD^{m\times (m+n)}
~~~\mbox{for}~~~ P\in\FD^{m\times m}.
\]
Then by \cite[Theorem 4.5, 5.5]{DaviesChengLabahn:2008},
$P\in\FD^{m\times m}$ is unimodular and $PA$ is in Popov form.
In particular, only $r$ rows are non-zero.

\begin{thm}
  Let $A\in\k[z][\D;\sigma,\delta]^\nxn$ have rank $r\leq m$, with
  $\degD A\leq d$ and $\deg_z A\leq e$.  We can find a unimodular
  matrix $P\in\k[z][\D;\sigma,\delta]^{m\times m}$ such that $PA$ has
  $r$ non-zero rows using $O(m^9n^9 (\degD A)^4 (\deg_z A)^4)$
  operations in $\k$.  It also requires a polynomial number of bit
  operations when $\k=\QQ$.
\end{thm}
\begin{proof}
  The matrix $B$ has $\degD(B)\leq (m+1)\cdot \degD A$ and
  $\deg_z(B)\leq \deg_z(A)$.  The algorithm of
  \cite{BeckermannChengLabahn:2006} to compute the Popov form,
  as summarized in their Corollary 7.7, requires $O(m^9n^9 (\degD A)^4
  (\deg_z A)^4)$ operations in $\k$.  It also requires a polynomial
  number of
  bit operations when $\k=\QQ$.
\end{proof}

After eliminating the zero rows, we are left to compute the Hermite
form of a (possibly rectangular) matrix with full rank.

\subsection{Rectangular matrices of full rank}

Let $A\in\FD^{m\times n}$ have full rank with $n>m$.  Then there
exists a lexicographically first set of columns $\tau_1,\ldots,\tau_m$
such that the submatrix
\[
A_\tau = 
\begin{pmatrix}
A_{1,\tau_1} & A_{1,\tau_2} & \cdots & A_{1,\tau_m}\\
\vdots &  &  & \vdots\\
A_{m,\tau_1} & A_{m,\tau_2} & \cdots & A_{m,\tau_m}
\end{pmatrix}
\in\FD^{m\times m}
\]
has full rank.  We can compute the unique $U$ such that $UA_\tau$
is in Hermite form using the algorithm of the previous section.  Then
it must be the case that $UA$ has form
\[
UA = 
\begin{pmatrix}
   & H_{1,\tau_1} & * & \cdots & \cdots & \cdots & * \\
   &            &   & H_{2,\tau_2} & * & \cdots & *\\
   &            &   &             &   & \ddots &  & \\
   &            &   &             &   &        & H_{m,\tau_m}
\end{pmatrix}\in\FD^{m\times n},
\]
the Hermite form of $A$.  This suggests an easy strategy of simply
conducting a binary search for the lexicographically smallest subset
$\{\tau_1,\ldots,\tau_m\}$ of $\{1,\ldots,n\}$ such that the $U$
computed to put $A_\tau$ in Hermite form also puts $UA$ in Hermite
form.  Since there are $\binom{n}{m}\leq 2^n$ subsets of size $m$, our
binary search will require at most $\log_2\binom{n}{m}\leq n$
iterations.

We can now offer the following theorem for rectangular matrices of
full row rank.

\begin{thm}
  Let $A\in\k[z][\D;\sigma,\delta]^{m\times n}$ have full rank with
  entries of degree at most $d$ in $\D$, and of degree at most $e$ in
  $z$.  Let $H\in\k(z)[\D;\sigma,\delta]^{m\times n}$ be the Hermite form of
  $A$ and $U\in\k(z)[\D;\sigma,\delta]^{m\times m}$ such that $UA=H$.
  \begin{enumerate}
  \item[(a)] $\deg_z H_{ij} \in O(m^2de)$  for $1\leq i,j\leq n$, and $\deg_z U_{ij}\in
    O(m^2de)$ for $1\leq i,j\leq m$.
  \item[(b)] We can compute $H$ and $U$ with a deterministic algorithm
    that requires
    $O(nm^7d^4$ $\log(md)\cdot\M(m^2de)+n^2m^3d^2\cdot\M(m^2de))$ or
    $\softO(nm^9d^3e+n^2m^5d^3e)$ operations in $\k$.
  \item[(c)] Assume $\k$ has at least $4m^2de$ elements.  We can
    compute $H$ and $U$ with an expected number
    $O(nm^7d^3\log(md)+nm^7d^3e+n^2m^5d^3e)$ of operations in $\k$
    (using standard polynomial arithmetic).  This algorithm is
    probabilistic of the Las Vegas type; it never returns an incorrect
    answer.
    \end{enumerate}
\end{thm}
\begin{proof}
  Part (a) from Theorem \ref{thm:sqhermcost} (a), follows since the
  transformation is computed from an $m\times m$ submatrix of $A$.

  Part (b) follows because each iteration of the binary search
  requires computing the Hermite form and transformation matrix $U$ of
  an $m\times m$ submatrix of $A$.  The cost of this is shown in
  Theorem \ref{thm:sqhermcost}.  We then check if $UA$ is in Hermite
  form, which is a matrix multiplication of the $m\times m$ matrix $U$
  times the $m\times n$ matrix $A$.  Each entry in $U$ has degree
  $O(md)$ in $\D$ and degree $O(m^2de)$ in $z$, so the check requires
  $O(nm^3d^2\cdot\M(m^2de))$.  Both the Hermite form computation and
  the check are done $n$ times, giving the total shown cost.
  Note that in the event we choose an $m\times m$ submatrix which is row rank
  deficient, the algorithm will fail, either in the computation of the
  Hermite form, or in the final multiplication test, and we do not
  need to resort to the more expensive row-reduction algorithm
  outlined in Subsection \ref{ssec:nonfullrank}.

  Part (c) follows similarly to part (b), except that the
  probabilistic algorithm described in Theorem \ref{thm:sqhermcost}
  (c) is used.
\end{proof}


  \section*{Acknowledgement}

  The authors would like to thank Howard Cheng for his assistance with
  Section \ref{ssec:nonfullrank}, and the anonymous referees for
  their exceptionally detailed and helpful reviews.

  They also acknowledge the support of the Natural Sciences and
  Engineering Research Council of Canada (NSERC) and MITACS Canada.
  
   \newcommand{\Gathen}{\relax}\newcommand{\Hoeven}{\relax}

\end{document}